\documentclass{article}
\usepackage[utf8]{inputenc}
\usepackage{amsmath}
\usepackage{amssymb}
\usepackage{amsthm}

\newtheorem{lemma}{Lemma}
\newtheorem{proposition}{Proposition}
\newtheorem{theorem}{Theorem}
\newtheorem{corollary}{Corollary}
\newtheorem{remark}{Remark}

\DeclareMathOperator{\CFL}{CFL} 
\DeclareMathOperator{\REG}{REG} %

\title{Note on dissecting power of regular languages}
\author{Josef Rukavicka}
\date{October 2023}

\begin{document}

\maketitle

\begin{abstract}
Let $c>1$ be a real constant. We say that a language $L$ is $c$-\emph{constantly growing} if for every word $u\in L$ there is a word $v\in L$ with $\vert u\vert<\vert v\vert\leq c+\vert u\vert$.
We say that a language $L$ is $c$-\emph{geometrically growing} if for every word $u\in L$ there is a word $v\in L$ with $\vert u\vert<\vert v\vert\leq c\vert u\vert$.
Given a language $L$, we say that $L$ is $\REG$-\emph{dissectible} if there is a regular language $R$ such that  $\vert L\setminus R\vert=\infty$ and $\vert L\cap R\vert=\infty$. In 2013, it was shown that every $c$-constantly growing language $L$ is $\REG$-dissectible. In 2023, the following open question has been presented: ``Is the family of geometrically growing languages $\REG$-dissectible?'' 

We construct a $c$-geometrically growing language $L$ that is not $\REG$-dissectible. Hence we answer negatively to the open question.
\end{abstract}
x
\section{Introduction}

Let $L_1$ and $L_2$ be two infinite languages. We say that $L_1$ \emph{dissects} $L_2$ if \[\vert L_2\setminus L_1\vert=\infty\quad\mbox{ and }\quad\vert L_1\cap L_2\vert=\infty\mbox{.}\] Dissection of infinite languages have been investigated in recent years. Most notably, in \cite{YAMAKAMI2013116}, the dissecting power of regular languages has been researched.

Let $\mathcal{C}$ be a family of languages. We say that a language $L_2$ is $\mathcal{C}$-\emph{dissectible}  if there is $L_1\in \mathcal{C}$ such that $L_1$ dissects $L_2$. Let $\REG$ denote the family of regular languages and let $\CFL$ denote the family of context free languages. In \cite{YAMAKAMI2013116}, several examples of language families that are $\REG$-dissectible  have been presented. In addition, it was proved that the language \[L=\{a^{n!}\mid n\mbox{ is a positive integer }\}\] is not $\REG$-dissectible.

An infinite language $L$ is called \emph{constantly growing} or $c$-\emph{constantly growing} if there is a real constant $c\geq 1$ such that for every word $u\in L$ there is a word $v\in L$ with $\vert u\vert<\vert v\vert\leq c+\vert u\vert$. In \cite{YAMAKAMI2013116}, it has been proved that every constantly growing language $L$ is $\REG$-dissectible.

In \cite{rukavicka_dmtcs_2023}, a ``natural''  generalization of constantly growing languages has been introduced as follows.  
A language $L$ is called a \emph{geometrically growing} or $c$-\emph{geometrically growing} language if there is a real constant $c>1$ such that for every $u\in L$ there exists $v\in L$ with $\vert u\vert<\vert v\vert\leq c\vert u\vert$. 

In \cite{rukavicka_dmtcs_2023} it was shown how to dissect a  geometrically growing language by a homomorphic image of intersection of two context-free languages.

In the current article, for every real constant $c>1$, we construct a $c$-geometrically growing language $L$ that is not $\REG$-dissectible. The language $L$ is constructed over an alphabet with one letter. 
This implies also that $L$ is not $\CFL$-dissectible; just realize that every context free language over one letter alphabet is a regular language. If we consider only regular languages, context free languages, and their intersections then, up to the homomorphic images, we can claim that the article \cite{rukavicka_dmtcs_2023} presents the ``best possible'' result, how to dissect geometrically growing languages.

Our result gives a negative answer to the open question from \cite{rukavicka_dmtcs_2023}:  ``Is the family of geometrically growing languages $\REG$-dissectible?''. 

For a more detailed overview and other open questions concerning the dissection of infinite languages, see \cite{rukavicka_dmtcs_2023} and \cite{YAMAKAMI2013116}.

\section{Preliminaries}
Let $\mathbb{N}_0$ denote the set of all non-negative integers, 
let $\mathbb{N}_1$ denote the set of all positive integers and 
let $\mathbb{R}^+$ denote the set of all positive real numbers.

Let $\Sigma´=\{a\}$ denote the alphabet containing exactly one letter $a$.

\section{Geometrically growing language}
Suppose $\alpha,\beta\in\mathbb{N}_1$ with $\alpha<\beta$. Let $\gamma=\frac{1}{\ln{\beta}-\ln{\alpha}}$.

Let \[\Delta=\{(j,n)\mid j,n\in\mathbb{N}_0\mbox{ and }n>\alpha\lfloor\gamma\ln{n}\rfloor\mbox{ and }j\in\{0,1,\dots,\lfloor\gamma\ln{n}\rfloor\}\}\mbox{.}\]
Given $(j,n)\in\Delta$, let \[\phi(j,n)=\left(\frac{\beta}{\alpha}\right)^{j}n!\] and let \[\omega(n)=\frac{n!}{(\lfloor\gamma\ln{n}\rfloor \alpha+1)!}\mbox{.}\]

We will use the function $\phi(j,n)$ to define the lengths of words. Hence we need the value of $\phi(j,n)$ to a positive integer; the next lemma shows that this requirement is satisfied. In addition we show that $\phi(j,n)$ is divisible by $\omega(n)$.

\begin{lemma}
    \label{fjf88e7f}    
    If $(j,n)\in\Delta$ then $\phi(j,n)$ is a positive integer and \[0\equiv \phi(j,n)\pmod{\omega(n)}\mbox{.}\]
\end{lemma}
\begin{proof}
From $n>\alpha\lfloor\gamma\ln{n}\rfloor$ we have that 
\[n!= n(n-1)\cdots (\lfloor\gamma\ln{n}\rfloor \alpha) \cdots ((\lfloor\gamma\ln{n}\rfloor -1)\alpha)\cdots (2\alpha)\cdots \alpha\cdots 1\mbox{.}\]
It follows that $n!=m\alpha^{\lfloor\gamma\ln{n}\rfloor}$ for some $m\in\mathbb{N}_1$.
Consequently $\phi(j,n)$ is a positive integer.

This ends the proof.
\end{proof}

We show that for given $(j,n)\in\Delta$ there is $(\overline j,\overline n)\in\Delta$ such that $\phi(\overline j,\overline n)$ is strictly bigger that $\phi(j,n)$ and that $\phi(\overline j,\overline n)$ and $\phi(j, n)$ are ``close'' to each other.
\begin{proposition}
\label{ffh8dhjefd}
    If $(j,n)\in\Delta$ then there is $(\overline j,\overline n)\in \Delta$ such that $\overline n\geq n$ and \[\phi(j,n)<\phi(\overline j,\overline n)\leq \frac{(n+1)\beta}{n\alpha}\phi(j,n)\mbox{.}\]
\end{proposition}
\begin{proof}
    If $j<\lfloor\gamma\ln{n}\rfloor$ 
    then \begin{equation}\label{kdiej98}\begin{split}
        \frac{\phi(j+1,n)}{\phi(j,n)}=\frac{\left(\frac{\beta}{\alpha}\right)^{j+1} n!}{\left(\frac{\beta}{\alpha}\right)^j n!}=\frac{\beta}{\alpha}\mbox{.}\end{split}
    \end{equation}

    We have that \begin{equation}\label{fkk9ejdd}\begin{split}\left(\frac{\beta}{\alpha}\right)^{\lfloor\gamma\ln{n}\rfloor} =\left(\frac{\beta}{\alpha}\right)^{\lfloor\frac{\ln{n}}{\ln{\beta}-\ln{\alpha}}\rfloor}\leq\\ \left(\frac{\beta}{\alpha}\right)^{\frac{\ln{n}}{\ln{\beta}-\ln{\alpha}}}=e^{\frac{\ln{n}}{\ln{\beta}-\ln{\alpha}}\ln{\frac{\beta}{\alpha}}}=e^{\ln{n}}=n\mbox{.}\end{split}\end{equation}
   
From (\ref{fkk9ejdd}) it follows that 
    \begin{equation}\label{kdiej99}\begin{split}
    \frac{\phi(0,n+1)}{\phi(\lfloor\gamma\ln{n}\rfloor,n)}=\frac{\left(\frac{\beta}{\alpha}\right)^{0}(n+1)!}{\left(\frac{\beta}{\alpha}\right)^{\lfloor\frac{\ln{n}}{\ln{\beta}-\ln{\alpha}}\rfloor}n!}=\frac{(n+1)!}{\left(\frac{\beta}{\alpha}\right)^{\lfloor\frac{\ln{n}}{\ln{\beta}-\ln{\alpha}}\rfloor}n!}\geq\frac{(n+1)}{n}>1\mbox{}
    \end{split}\end{equation}
    and \begin{equation}\label{kde36eddj}\begin{split}
    \frac{\phi(0,n+1)}{\phi(\lfloor\gamma\ln{n}\rfloor,n)}=\frac{\left(\frac{\beta}{\alpha}\right)^{0}(n+1)!}{\left(\frac{\beta}{\alpha}\right)^{\lfloor\frac{\ln{n}}{\ln{\beta}-\ln{\alpha}}\rfloor}n!}<\frac{(n+1)!}{\left(\frac{\beta}{\alpha}\right)^{\frac{\ln{n}}{\ln{\beta}-\ln{\alpha}}-1}n!}=\\\frac{(n+1)!}{n\left(\frac{\beta}{\alpha}\right)^{-1}n!}=\frac{\beta(n+1)}{\alpha n}\mbox{.}
    \end{split}\end{equation}
    
    We distinguish the following three cases when proving the proposition:
    \begin{itemize}
        \item If $0\leq j<\lfloor\gamma\ln{n}\rfloor$ then from (\ref{kdiej98}) it follows that \[\phi(j,n)<\phi(j+1,n)\leq \frac{\beta}{\alpha}\phi(j,n)\mbox{.}\]
        \item If $j=\lfloor\gamma\ln{n}\rfloor$ then from (\ref{kdiej99}) and (\ref{kde36eddj}) we have that \[\phi(j,n)<\phi(0,n+1)\leq \frac{(n+1)\beta}{n\alpha}\mbox{.}\]
    \end{itemize}
   It follows that for every $(j,n)\in\Delta$ there is $(\overline j,\overline n)\in\Delta$ such that \[\phi(j,n)<\phi(\overline j,\overline n)\leq \frac{(n+1)\beta}{n\alpha}\mbox{.}\]
    This ends the proof.
\end{proof}
\begin{remark}
    From the proof of Proposition \ref{ffh8dhjefd} it follows also that \[\phi(j,n)\leq \phi(\overline j,\overline n)\] if and only if $n\leq \overline n$ or ($n= \overline n$ and $j\leq \overline j$).
\end{remark}

Given the letter $a\in\Sigma$, let \[\begin{split}\Pi=\{a^{\phi(j,n)}\mid (j,n)\in\Delta\}\mbox{.}\end{split}\]

\begin{remark}
    Lemma \ref{fjf88e7f} implies that $\phi(j,n)$ is a positive integer. Hence the definition of the language $\Pi$ makes sense. Also it is clear that if $a^{\phi(j,n)}\in\Pi$ then $\vert a^{\phi(j,n)}\vert=\phi(j,n)$.
\end{remark}

We show that for any real constant $c>\frac{\beta}{\alpha}$ there is $c$-geometrically growing language that we construct from the language $\Pi$ by removing a finite number of its elements.
\begin{corollary}
\label{dh787e9fjrer}
If $c\in\mathbb{R}^+$ and $c>\frac{\beta}{\alpha}$ then there is $\overline \Pi\subseteq \Pi$ such that  
     the language $\overline \Pi$ is $c$-geometrically growing and $\vert\Pi\setminus\overline\Pi\vert<\infty$.
\end{corollary}
\begin{proof}
Let $n_0\in\mathbb{N}_1$ be such that \[\frac{(n_0+1)\beta}{n_0\alpha}<c\mbox{.}\]
Clearly such $n_0$ exists and also for every $n>n_0$ we have that \[\frac{(n+1)\beta}{n\alpha}<c\mbox{.}\]

Proposition \ref{ffh8dhjefd} asserts that for every $(j,n)\in\Delta$ with $n>n_0$ there is $(\overline j,\overline n)\in\Delta$ such that  \begin{equation}\label{fgfrj99dg6y}\overline n\geq n\quad\mbox{ and }\quad
\phi(j,n)<\phi(\overline j,\overline n)\leq \frac{(n+1)\beta}{n\alpha}\phi(j,n)<c\phi(j,n)\mbox{.}\end{equation}

Let \[\overline \Pi=\{a^{\pi(j,n)}\in\Pi\mid n>n_0\}\mbox{.}\]
From (\ref{fgfrj99dg6y}) it follows that for every $a^{\phi(j,n)}\in\overline \Pi$ there is $a^{\phi(\overline j,\overline n)}\in\overline \Pi$ such that \[\vert a^{\phi(j,n)}\vert<\vert a^{\phi(\overline j,\overline n)}\vert\leq c\vert a^{\phi(j,n)}\vert\mbox{.}\]
We conclude that $\overline \Pi$ is a $c$-geometrically growing language.

This ends the proof.
\end{proof}

\section{Dissecting by a regular language}
\begin{remark}
    We suppose the reader to be familiar with the regular languages and deterministic finite automata. As such, we present the proofs of Lemma \ref{jduiej9ff} and Lemma \ref{idrr9rf} in a less formal way.
\end{remark}
Given a deterministic fnite automaton $D$, let $\mathcal{L}(D)\subseteq \Sigma^*$ denote the set of words that $D$ accepts. 
Let \[\begin{split}\REG(1)=\{R \subseteq\Sigma^*\mid \vert R\vert=\infty\mbox{ and there is a deterministic finite }\\\mbox{ automaton }D \mbox{ such that }\mathcal{L}(D)=R\\\mbox{ and }D\mbox{ has exactly one accepting state}\}\mbox{.}\end{split}\] 
Less formally said, $\REG(1)$ is the set of all infinite regular languages that can be accepted by a deterministic finite automaton having exactly one accepting state. The next lemma illuminates the properties of such regular languages.
\begin{lemma}
\label{jduiej9ff}
    If $R\in\REG(1)$ then there are $q\in\mathbb{N}_0$ and  $r\in\mathbb{N}_1$  such that \[R=\{a^{q+ir}\mid i\in\mathbb{N}_1\}\mbox{.}\]
\end{lemma}
\begin{proof}
    Consider a deterministic finite automaton $D$ accepting $R$ such that $D$ has exaclty one accepting state. Since $R\in\REG(1)$ we have that such $D$ exists. 
    Since the alphabet $\Sigma$ has only one letter, it is easy to see that $D$ can be represented as a directed graph with $q$+$r$ vertices, where $r$ vertices form a cycle. One vertice in the cycle represents the single accepting state of $D$.
    The lemma follows. 

This ends the proof.
\end{proof}

We show that when investigating a dissecting power of regular languages on one letter alphabet, it suffices to consider regular languages from $\REG(1)$.
\begin{lemma}
\label{idrr9rf}
    If a regular language $R\subseteq \Sigma^*$ dissects a language $L\subseteq \Sigma^*$ then there is $\overline R\in\REG(1)$ such that $\overline R$ dissects $L$.
\end{lemma}
\begin{proof}
Let $D$ be a deterministic finite automaton accepting $R$ and let $T$ be the set of accepting states of $D$. Given an accepting state $t\in T$, we define \[R(t)=\{u\in R\mid \mbox{ Given the input }u\mbox{, } D\mbox{ halts in the state }t\}\mbox{.}\]
Since $\vert R\vert=\infty$, it is clear that if $R$ dissects $L$ then obviously there is at least one accepting state $t\in T$ such that $\vert R(t)\cap L\vert=\infty$.  Because $R(t)\subseteq R$, it follows that $R(t)$ dissects $L$. 
The definition of $R(t)$  implies that $R(t)\in\REG(1)$. The lemma follows.

This ends the proof.
\end{proof}

Equipped with Lemma \ref{jduiej9ff} and Lemma \ref{idrr9rf} we can prove that there is no regular language $R$ such $R$ dissects $\Pi$.
\begin{proposition}
\label{hhbd7egdj}
The language $\Pi$ is not $\REG$-dissectible.
\end{proposition}
\begin{proof}
Suppose that the language $\Pi$ is $\REG$-dissectible. Then Lemma \ref{jduiej9ff} and Lemma \ref{idrr9rf} imply that there are $q,r\in\mathbb{N}_0$ and $R\in\REG(1)$ such that \[R=\{a^{q+ir}\mid i\in\mathbb{N}_1\}\] and $R$ dissects $\Pi$. 

Let $n_0$ be such that $n_0-\alpha\lfloor\gamma\ln{n_0}\rfloor>r$. Realize that \[\lim_{n\rightarrow\infty}(n-\alpha\lfloor\gamma\ln{n}\rfloor)=\infty\mbox{,}\] hence such $n_0$ exists. 
Obviously $n-\alpha\lfloor\gamma\ln{n}\rfloor>r$ for all $n>n_0$. It follows that there is \[m\in\{\alpha\lfloor\gamma\ln{n}\rfloor+1, \alpha\lfloor\gamma\ln{n}\rfloor+2, \dots,n\}\] such that $0\equiv m\pmod{r}$. In consequence, we have that 
\[0\equiv\omega(n)\pmod{r}\mbox{.}\] 
Lemma \ref{fjf88e7f} implies that \[0\equiv\phi(j,n)\pmod{\omega(n)}\mbox{.}\]
It follows that if $n>n_0$ then \[0\equiv\phi(j,n)\pmod{r}\mbox{.}\]

Hence there are $\overline \Pi\subseteq \Pi$ and $h\in\mathbb{N}_0$ such that $\vert \Pi\setminus\overline \Pi\vert<\infty$, $h\geq q$ and \[\overline \Pi\subseteq \{a^{h+ir}\mid i\in\mathbb{N}_1\}\mbox{.}\] 
We distinguish two following cases:
\begin{itemize}
\item If $0\equiv q\pmod{r}$ then $\overline \Pi\subseteq R$.
\item If  $0\not\equiv q\pmod{r}$ then $\overline \Pi\cap R=\emptyset$.
\end{itemize}
We conclude that $\Pi$ is not $\REG$-dissectible. 

This ends the proof.
\end{proof}

Now we step to the main result of the current article.
\begin{theorem}
For every $c\in\mathbb{R}^+$ with $c>1$ there is a $c$-geometrically growing language $L$ that is not $\REG$-dissectible.
\end{theorem}
\begin{proof}
Let $\alpha,\beta\in\mathbb{N}_1$ be such that $c>\frac{\beta}{\alpha}$ and $\beta>\alpha$. Since $c>1$ we have that such $\alpha,\beta$ exist. 
Then Corollary \ref{dh787e9fjrer} asserts that there is $\overline \Pi\subseteq \Pi$ such that  
     $\overline \Pi$ is $c$-geometrically growing and $\vert\Pi\setminus\overline\Pi\vert<\infty$. 
From Proposition \ref{hhbd7egdj} we have that $\Pi$ is not $\REG$-dissectible. Then it is easy to see that for every $S\subseteq \Pi$ with $\vert\Pi\setminus S\vert<\infty$ we have that $S$ is not $\REG$-dissectible. 
Hence we conclude $\overline\Pi$ is not $\REG$-dissectible.

 This completes the proof.
\end{proof}

\bibliographystyle{siam}
\IfFileExists{biblio.bib}{\bibliography{biblio}}{\bibliography{../!bibliography/biblio}}

\end{document}